%% file: root.tex
\documentclass[letter, 10pt, conference]{ieeeconf}      
\makeatletter
\def\endthebibliography{%
  \def\@noitemerr{\@latex@warning{Empty `thebibliography' environment}}%
  \endlist
}
\makeatother

\IEEEoverridecommandlockouts 
\usepackage{times} 
\usepackage{amsmath} 
\usepackage{mathrsfs}  
\usepackage{bm}
\usepackage{amssymb}  
\usepackage{xcolor}
\usepackage{optidef}
\usepackage{glossaries}
\usepackage{booktabs}
\usepackage[font=small]{caption}
\usepackage{subcaption}

\usepackage[compress,noadjust]{cite}

\makeatletter
\let\NAT@parse\undefined
\makeatother
\usepackage[pagebackref=false,breaklinks=true,colorlinks=true,bookmarks=false,allcolors=black,citecolor=black]{hyperref}

\usepackage{graphicx}
\graphicspath{{figures/}}
\usepackage[acronym,style=super,nogroupskip,nonumberlist,toc=false]{glossaries-extra}
\setabbreviationstyle[acronym]{long-short}
\loadglsentries{acronyms}
\usepackage{cleveref}
\Crefname{figure}{Fig.}{Figs.}
\usepackage{physics}
\usepackage[mode=text]{siunitx}
\usepackage{todonotes}

\usepackage[utf8]{inputenc}
\usepackage[T1]{fontenc}
\usepackage[USenglish]{babel}

\usepackage{amsthm}

\theoremstyle{plain}

\newtheorem{problem}{\textbf{Problem}}

\newtheorem{definition}{\bf{Definition}}
\newtheorem{lemma}{\bf{Lemma}}

\newtheorem{remark}{\bf{Remark}}

\renewcommand{\vec}[1]{\boldsymbol{\mathbf{#1}}}

\newcommand{\mat}[1]{\boldsymbol{\mathbf{#1}}}

\title{\LARGE \bf
Corridor MPC for Multi-Agent Inspection of Orbiting Structures \vspace{-4mm}
}

\author{Gregorio Marchesini\textsuperscript{\dag}, Pedro Roque\textsuperscript{\dag}, and Dimos V. Dimarogonas\textsuperscript{\dag}
\vspace{-16mm}
\thanks{\textsuperscript{\ddag}G.Marchesini, P. Roque, and D. V. Dimarogonas are with the Division of Decision and Control Systems, KTH Royal Institute of Technology, Stockholm, Sweden. 
        E-Mail: {\tt\small $\{$gremar, padr, dimos$\}$@kth.se}
        }
\thanks{This work was supported by the H2020 ERC Grant LEAFHOUND, the Horizon Europe EIC project SymAware (101070802), the Swedish Research Council (VR), the Knut och Alice Wallenberg Foundation (KAW), and the Wallenberg AI, Autonomous Systems and
Software Program (WASP) DISCOWER funded by the Knut and Alice Wallenberg Foundation.}
}

\makeatletter
\def\footnoterule{\relax%
  \kern-5pt
  \hbox to \columnwidth{\vrule width 0.5\columnwidth height 0.4pt\hfill}
  \kern4.6pt}
\makeatother

\input{costumeSymbols}

\begin{document}

\maketitle
\thispagestyle{empty}
\pagestyle{empty}

\input{content/abs.tex}
%
\section{INTRODUCTION}
\input{content/intro.tex}
\input{content/notation.tex}
\section{BACKGROUND}
\input{content/background.tex}
\section{PROBLEM STATEMENT}
\input{content/problemStatement.tex}

\section{SAMPLED DATA HOCBF}
\input{content/sampleDataHOCBF}

\section{CONTROL STRATEGY}
\input{content/control.tex}
\section{RESULTS}
\input{content/results.tex}

\section{CONCLUSIONS AND FUTURE WORK}
\input{content/conclusion.tex}

{
\bibliographystyle{IEEEtran}  
\bibliography{references}  
}

\end{document}

%% file: acronyms.tex
\newacronym[description={model predictive control(ler)}]{mpc}{MPC}{model predictive control}
\newacronym{com}{CoM}{center of mass}
\newacronym{uav}{UAV}{unmanned aerial vehicle}
\newacronym{rhc}{RHC}{receding horizon control}
\newacronym{pid}{PID}{proportional-integral-derivative}
\newacronym{lq}{LQ}{linear quadratic}
\newacronym{lqr}{LQR}{linear quadratic regulator}
\newacronym{dare}{DARE}{discrete algebraic ricatti equation}
\newacronym{dae}{DAE}{differential-algebraic equation}
\newacronym{ode}{ODE}{ordinary differential equation}
\newacronym{dp}{DP}{dynamic programming}
\newacronym{lp}{LP}{linear program}
\newacronym{nlp}{NLP}{nonlinear program}
\newacronym{qp}{QP}{quadratic program}
\newacronym{zoh}{ZOH}{zero-order hold}
\newacronym{rk4}{RK4}{fourth order Runge--Kutta}
\newacronym{lti}{LTI}{linear time-invariant}
\newacronym{jit}{JIT}{just-in-time}
\newacronym{dof}{DOF}{degrees of freedom}
\newacronym{api}{API}{application programming interface}
\newacronym{ggl}{GGL}{Gear--Gupta--Leimkuhler}
\newacronym{ros}{ROS}{Robot Operating System}
\newacronym{fcu}{FCU}{flight control unit}

%% file: costumeSymbols.tex
\newcommand{\csf}{\mathcal{C}_{S}}

\newcommand{\dr}{\delta \bm{r}}

\newcommand{\dvel}{\delta \bm{v}}

\newcommand{\rdir}{\hat{\bm{r}}}
\newcommand{\sdir}{\hat{\bm{s}}}
\newcommand{\wdir}{\hat{\bm{w}}}


%% file: content/abs.tex
\begin{abstract}
In this work, we propose an extension of the previously introduced Corridor Model Predictive Control scheme for high-order and distributed systems, with an application for on-orbit inspection. To this end, we leverage high order control barrier function (HOCBF) constraints as a suitable control approach to maintain each agent in the formation within a \textit{safe} corridor from its reference trajectory. The recursive feasibility of the designed MPC scheme is tested numerically, while suitable modifications of the classical HOCBF constraint definition are introduced such that safety is guaranteed both in sampled and continuous time. The designed controller is validated through computer simulation in a realistic inspection scenario of the International Space Station.
\end{abstract}


%% file: content/intro.tex
\label{introduction}
The application of multi-agent systems (MAS) design to solve complex robotics tasks has received increasing attention in the past few decades 
 \cite{dorri_multi-agent_2018,mesbahi_graph_2010}. Examples of successful MAS control paradigms for terrestrial and aerial applications are extensive in the literature due to the broad application range. The advantages of MAS design include redundancy and robustness to single agent failure, reduced complexity in single agent hardware and the possibility to accomplish complex interactions among heterogeneous agents. These same advantages are of critical importance for the next generation of planetary exploration, on-orbit servicing and construction mission concepts, to mention a few \cite{ekblaw_self-assembling_2019, khoshnevis_isru-based_2017, flores-abad_review_2014}.\par
In this work, we propose a solution to the problem of multi-agent inspection of on-orbit space vehicles using unmanned autonomous spacecrafts \cite{nakka_information-based_2021}. The ability to autonomously inspect space vehicles has the potential to lower the cost of replacing space assets and hence help in reducing the population of space debris orbiting the Earth. This fact drives our work towards a fully safe and autonomous inspection of such space assets. We structure the inspection mission as follows: we assume a formation of CubeSats, called the \textit{inspectors} \cite{nakka_information-based_2021}, is deployed into a set of Passive Relative Orbits (PRO) \cite{alfriend_spacecraft_2009} around a space vehicle orbiting a planetary body in a nearly circular orbit. Each inspector is controlled through a sampled-data model predictive controller (MPC), which is applied to track the assigned PRO under the influence of orbital perturbations. Appropriate High Order Control Barrier Function 
(HOCBF) constraints are introduced within the MPC scheme to constrain the system inside worst-case velocity and position tracking error bounds for each inspector in the formation.\par
The application of sampled-data CBF inside a Finite Horizon Optimal Control scheme (FHOC) such as MPC has already been explored in \cite{zeng_safety-critical_2020}. Still, safety in between discrete time steps is not analyzed. This problem is first addressed in \cite{cortez_control_2019} and \cite{breeden_control_2021}, where suitable corrective terms are added to the continuous time CBF constraint formulation to ensure safety between time steps. However, only first relative degree systems under zero disturbances are analyzed in \cite{breeden_control_2021}, while only time-invariant dynamics are analyzed in \cite{cortez_control_2019}. In \cite{roque_corridor_2022}, the results from \cite{cortez_control_2019,breeden_control_2021} are unified under a unique framework and expanded to first-order relative degree systems with time-varying dynamics and subject to state disturbances.\par
Based on the Corridor MPC (CMPC) scheme developed by \cite{roque_corridor_2022}, the contributions of this work are as follows: i) expand the definition of sampled-data CBF in \cite{roque_corridor_2022} to the case of higher relative degree systems, ii) apply the expanded CMPC control scheme to a realistic multi-agent inspection mission of a space vehicle. The recursive feasibility of the derived CMPC control scheme is then shown numerically following an approach similar to \cite{tan_compatibility_2022}.
We remark that only the planning and control parts of the inspection mission are analyzed here, while problems inherent to the visual inspection of the space vehicle are topics of future work.\par
The manuscript is divided as follows: Section \ref{background} reviews the fundamentals of relative spacecraft dynamics and HOCBF for safety-critical systems. In Section \ref{problem-statement}, we formally present the inspection problem, and Section \ref{discrete time barrier function section} proposes the new definition of sample data HOCBF as an expansion to the work in \cite{roque_corridor_2022,cortez_control_2019,breeden_control_2021}. Section \ref{control} presents the Corridor MPC for high-order systems. Lastly, Sections \ref{results}-\ref{conclusion} show a numerical simulation proving the applicability of the proposed solution in a realistic inspection of the International Space Station (ISS), followed by the conclusions.

%% file: content/notation.tex
\par
\textit{Notation:} Small, bold letters represent vectors. 
Matrices are denoted by bold, capital letters. 
Regular letters denote scalars.
Calligraphic letters denote reference frames, and the basis vectors of a frame $\mathcal{A}$ are denoted $\{\vec{a}_x,\vec{a}_y,\vec{a}_z\}$.
The weighted vector norm $\sqrt{\vec x^T \mat A \vec x}$ is denoted $\norm{\vec x}_{\mat A}$. 
The notation $\| \cdot \|$ represents the standard Euclidean 2-norm, while the hat symbol ($\hat{\cdot}$) on top of a vector quantity denotes a unitary (unit-norm) vector.
A continuous function, $\alpha(\cdot):\mathbb{R}\rightarrow\mathbb{R}$ is an extended class $\mathcal{K}$-function if it
is strictly increasing and  $\alpha(0) = 0$, while $\alpha(\cdot)$ is of class $C^{l}$ if it is $l$ times continuously differentiable in its argument.
Given a sampling time $\Delta t \geq 0$, we define a discrete time instant as $k\Delta t\triangleq t_0+k\Delta t$ with $k\in\mathbb{N}_{0}=\mathbb{N}\cup\{0\}$ and initial reference time $t_0\in\mathbb{R}_{\geq0}$.
We will use the notation $\bm{a}(i|k\Delta t)$ to indicate a property that is predicted $i$-steps ahead relative to the current discrete time instant $k\Delta t$.  
We denote vectorial/scalar properties of the space vehicle and the inspectors with the subscripts $sv$ and $ins$ respectively. 

%% file: content/background.tex
\label{background}
\subsection{Nonlinear Relative Dynamics}
\begin{figure}[!h]
    \vspace{-12mm}
    \centering
    \includegraphics[trim={2.6cm 3.4cm 2.3cm 1.8cm},clip,scale=0.29]{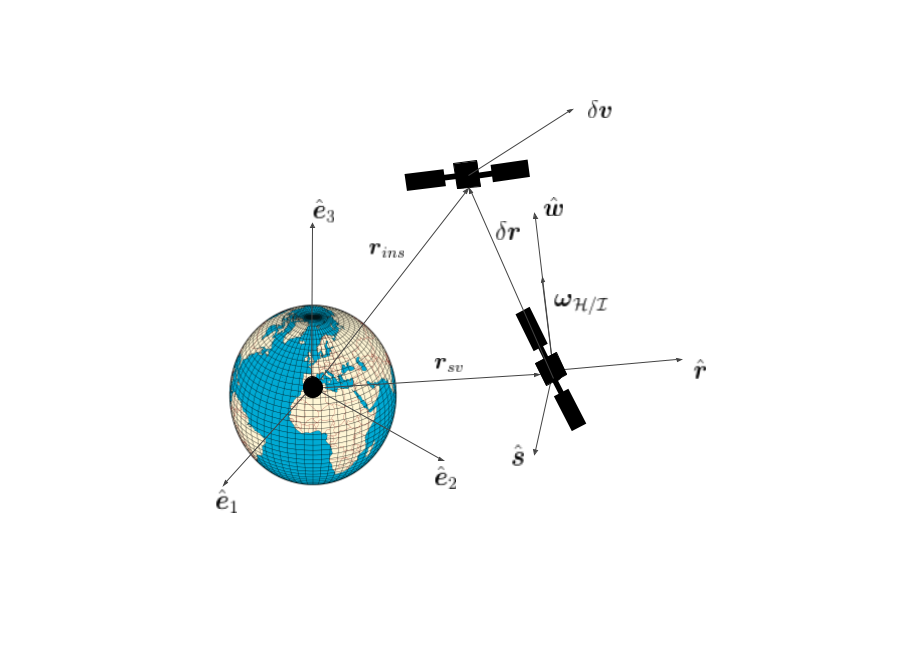}
    \vspace{-6mm}
    \caption{Relative state of the inspector spacecraft ($\delta \bm{r}$) with respect to the a general SV ($\bm{r}_{sv}$). The Hill's frame is defined by the base $\{\hat{\bm{r}},\hat{\bm{s}},\hat{\bm{w}}\}$ while the inertial frame $\mathcal{J}$ is defined by $\{\hat{\bm{e}}_1,\hat{\bm{e}}_2,\hat{\bm{e}}_3\}$.}
    \label{fig:hill frame}
    \vspace{-2mm}
\end{figure}
Consider $\mathcal{J}$ to be an inertial frame fixed at the Earth's centre of mass with base \{$\hat{\bm{e}}_1$,$\hat{\bm{e}}_2$,$\hat{\bm{e}}_3$\}, such that $\hat{\bm{e}}_3$ is aligned with the rotational axis of the Earth. Furthermore, consider the Local Vertical Local Horizontal frame $\mathcal{H}$ over the space vehicle's inertial state with base $\{\hat{\bm{r}},\hat{\bm{s}},\hat{\bm{w}}\}$ with $
\rdir = \frac{\bm{r}_{sv}}{\|\bm{r}_{sv}\|}$, $
\wdir = \frac{\bm{r}_{sv} \times \dot{\bm{r}}_{sv}}{\|\bm{r}_{sv} \times \dot{\bm{r}}_{sv}\|}$, $
\sdir = \wdir \times \rdir$, where $\bm{r}_{sv}$ and $\dot{\bm{r}}_{sv}$ are the inertial position and velocity of space vehicle with respect to $\mathcal{J}$. We define the relative position and velocity of the inspector in the $\mathcal{H}$ frame as $\delta \bm{r} =\bm{r}_{ins}-\bm{r}_{sv}$ and $\delta \bm{v}= \delta \dot{\bm{r}}-\bm{\omega}_{\mathcal{H}/\mathcal{J}}\times \delta \bm{r}$ respectively, with $\bm{\omega}_{\mathcal{H}/\mathcal{J}}$ being the angular velocity of $\mathcal{H}$ w.r.t to $\mathcal{J}$ and $\delta \dot{\bm{r}}$ being the time derivative of $\delta \bm{r}$ w.r.t $\mathcal{J}$. 
Given the inspector relative state $\bm{x}=[\delta \bm{r},\delta \bm{v}]^T \in \tilde{\mathbb{X}}\subset \mathbb{R}^{6}$ and time $t\in\mathbb{I}\subset \mathbb{R}_{\geq0}$, the nominal dynamics of an inspector relative to the space vehicle written in state space form is given as
\begin{equation}\label{eq:unperturbed nonlinear system}
\dot{\bm{x}}= \eta(\bm{x},t) = \bm{f}(\bm{x},t) + \bm{g}(\bm{x},t)\bm{u}\\
\end{equation}
with $\bm{g}(\bm{x},t) =[\bm{O}_{3}, \bm{I}_{3}]^T$, $\bm{f}(\bm{x},t) =[
\bm{f}_r(\bm{x},t), \bm{f}_v(\bm{x},t)]^T$ and $\bm{u}\in\mathbb{U}\subset\mathbb{R}^3$ the acceleration from the inspector's propulsion system. We assume $\mathbb{U}$ and $\tilde{\mathbb{X}}$ are compact sets containing the origin, and $\mathbb{I}$ is a compact time interval.  The full dynamics of \eqref{eq:unperturbed nonlinear system} can be found in \cite[Eqs. 4.14-4.16]{alfriend_spacecraft_2009} and are omitted for brevity. 
Nonetheless, $\bm{f}$ and $\bm{g}$ are of class $C^{\infty}$ in $t$ and $\bm{x}$ for any compact set $\tilde{\mathbb{X}}$ such that $\delta \bm{r}\neq \bm{r}_{sv} ;\forall \bm{x}\in\tilde{\mathbb{X}}$. Under the effect of orbital perturbations, the nominal dynamics \eqref{eq:unperturbed nonlinear system} become
\begin{equation}\label{eq:perturbed nonlinear system}
\dot{\tilde{\bm{x}}} = \tilde{\eta}(\tilde{\bm{x}},t,\bm{d}) = \bm{f}(\tilde{\bm{x}},t) + \bm{g}(\tilde{\bm{x}},t)(\bm{u}+\bm{d})
\end{equation}
where  $\eta(\bm{x},t) \triangleq\tilde{\eta}(\tilde{\bm{x}},t,\bm{0})$, $\bm{d}\in \mathbb{W}\subset\mathbb{R}^3$ is the orbital perturbation vector gathering $J_2$ and aerodynamic drag perturbations \cite{morgan_swarm-keeping_2012}, and $\tilde{\bm{x}}=[\delta \tilde{\bm{r}},\delta \tilde{\bm{v}}]\in\tilde{\mathbb{X}}\subset\mathbb{R}^6$ is the perturbed state. The tilde notation distinguishes the perturbed state dynamics from the nominal ones.
Although analytical models for $\bm{d}$ exist \cite{morgan_swarm-keeping_2012}, they depend on parameters that are difficult to identify online. Therefore, we consider $\bm{d}$ an unknown bounded disturbance encompassing noise. 
We further assume $\mathbb{W}$ a compact and convex set containing the origin such that 
\begin{equation}
\begin{gathered}
\label{eq:control and disturbances bounds}
\tilde{\mathbb{X}} = \{  \tilde{\bm{x}}\in \mathbb{R}^6: \| \tilde{\bm{x}}\| \leq \epsilon_{\tilde{x}}\},\mathbb{I}=[0,t_{max}],\\
\mathbb{U} = \{ \bm{u}\in \mathbb{R}^3: \|\bm{u}\| \leq \epsilon_u\},
\mathbb{W} = \{ \bm{d}\in \mathbb{R}^3: \|\bm{d}\| \leq \epsilon_d\}
\end{gathered}
\end{equation}
with $\epsilon_u,\epsilon_d,\epsilon_{\tilde{x}},t_{max}\in\mathbb{R}^{+}$.
We denote as $\eta_d(\bm{x},k\Delta t)$ the discrete-time version of the nominal dynamics $\eta(\bm{x},t)$, obtained through Runge-Kutta integration so that $
\bm{x}((k+1)\Delta t) = \eta_{d}(\bm{x}(k\Delta t),k\Delta t)$. As we restrict our derivations to SVs in nearly circular orbits, we use passive relative orbits \cite[Ch. 5]{alfriend_spacecraft_2009} as suitable reference trajectories to be tracked by each inspector. We refer to the reference PRO trajectory state at time $t$ as $\bm{x}_r(t)=[\delta \bm{r}_r(t),\delta \bm{v}_r(t)]$ according to Table \ref{tab:PRO table}.
\begin{table}[!htpb]
    \vspace{-2mm}
    \centering
    \begin{tabular}{clll}
    \toprule
     direction & $\delta \bm{r}_r(t)$&$\delta \bm{v}_r
     (t)$&$\delta \bm{a}_r(t)$\\
     \midrule
         $\hat{\bm{r}}$&  $\rho_{r} \sin {\omega_w^{r}}$ &$\omega_w\rho_{r}\cos {\omega_w^{r}}$ & $- \omega^2_w \rho_r\sin {\omega_w^r}$\\
         $\hat{\bm{s}}$&$\rho_{s}+2 \rho_{r} \cos {\omega_w^{r}}$& $-2\omega_w \rho_{r}\sin{\omega_w^{r}}$& $- 2\omega^2_w \rho_r\cos{\omega_w^r}$\\
         $\hat{\bm{w}}$&$\rho_{w} \sin{\omega_w^{w}}$&$\omega_w\rho_{w} \cos{\omega_w^{w}}$ & $-\omega^2_w \rho_w\sin{\omega_w^w}$\\
    \end{tabular}
    \caption{Component wise definition of a PRO orbit in \textit{amplitude-phase} parameters in the Local Vertical Local Horizontal frame $\mathcal{H}$. The shorthand notation $\omega_w^{r} \triangleq \omega_w t+\alpha_{r}$, $\omega_w^{w} \triangleq \omega_w t+\alpha_{w}$ is applied.}
    \label{tab:PRO table}
    \vspace{-3mm}
\end{table}

\noindent The parameters $\rho_r$, $\rho_s$, $\alpha_r$, and $\alpha_w$ are positive scalars and design parameters. The term $\omega_w$ is the mean motion of the space vehicle orbit, and it is given as $\omega_w=\sqrt{\mu/p^3}$ [\unit{s^{-1}}] with $p$ being the semi-major axis of the space vehicle orbit and $\mu$ being the standard gravitational parameters of the Earth. Such a reference trajectory can be proved to be an unforced solution of the Cloessy-Wiltshire (CW) model \cite[Ch. 5]{alfriend_spacecraft_2009}, which is obtained by linearising the nominal dynamics \eqref{eq:unperturbed nonlinear system} around the space vehicle inertial position and assuming the space vehicle's orbit to be perfectly circular (zero eccentricity) and under zero orbital perturbations. Assuming an inspector is correctly deployed into a PRO, due to the effect of orbital perturbations (considered absent in the CW model) and the nonlinearity of the real dynamics, the inspector would eventually drift from its assigned PRO if not properly stabilised. We hence define the position and velocity tracking errors under nominal dynamics as $\bm{e}_{\dr}(t) \triangleq \delta \bm{r}(t) - \delta \bm{r}_r(t)$, $\bm{e}_{\dvel}(t) \triangleq \delta \bm{v}(t) - \delta \bm{v}_r(t)$ and under perturbed dynamics as $\tilde{\bm{e}}_{\dr}(t) \triangleq\delta \tilde{\bm{r}}(t) - \delta \bm{r}_r(t) $, $\tilde{\bm{e}}_{\dvel}(t) \triangleq \delta \tilde{\bm{v}}(t) - \delta \bm{v}_r(t)$. Additionally, $\bm{e}(t)=[\bm{e}_{\dr},\bm{e}_{\dvel}]^T$ and $\tilde{\bm{e}}(t)=[\tilde{\bm{e}}_{\dr},\tilde{\bm{e}}_{\dvel}]^T$.
Without loss of generality, we will consider $\delta \bm{a}_{r}(t) = \bm{0}$.
\subsection{High Order Control Barrier Functions}\label{cbf section}
Control barrier functions (CBFs) \cite{ames_control_2019} and their high order version (HOCBFs) \cite{xiao_control_2019} are a commonly applied analytical tool to define control invariant sets where the system state can evolve \textit{safely} w.r.t a given notion of safety (\textit{i.e.} avoid an area where an obstacle is located or avoid a region of instability for the system). In this section, we review the key application of HOCBFs in the design of safe controllers.
Consider a continuously differentiable scalar function $h(\tilde{\bm{x}},t) : \mathbb{D} \times \mathbb{I} \rightarrow \mathbb{R}$, with $\mathbb{D} \subseteq \tilde{\mathbb{X}}$ being compact, and consider $\csf(t)$ and  $\partial \csf(t)$ to be the super level set of $h$ and its boundary, defined as $\csf(t) := \{\tilde{\bm{x}}\in \mathbb{D}: h(\tilde{\bm{x}},t) \geq 0\}$, $\partial \csf(t) := \{\tilde{\bm{x}}\in \mathbb{D}: h(\tilde{\bm{x}},t) = 0\}$. We recall the definition of relative degree as
\begin{definition}
(\cite[Def. 5]{xiao_control_2019})
The relative degree of a continuously differentiable function $h: \mathbb{D}\times \mathbb{I} \rightarrow \mathbb{R}$ with
respect to system equation \eqref{eq:perturbed nonlinear system} is the number of times it is  needed to
differentiate it along the dynamics of equation \eqref{eq:perturbed nonlinear system} until control $\bm{u}$ explicitly shows.
\end{definition}
When $h(\tilde{\bm{x}},t)$ is of relative degree 1, it is possible to directly define a valid input $\bm{u}$ such that $\csf(t)$ is control forward invariant as in \cite{ames_control_2019}. 
When the relative degree of $h$ is $r>1$, we can guarantee $\csf(t)$ is a control invariant set by rendering a strict subset of $\csf(t)$ forward invariant. Consider the cascade of functions $H_i : \mathbb{D} \times \mathbb{I} \rightarrow \mathbb{R}, \; \forall i=0,\ldots r$ as
\begin{equation}
\label{eq:barriers cascade}
\begin{aligned}
H_0(\tilde{\bm{x}},t) &= h(\tilde{\bm{x}},t);
H_i(\tilde{\bm{x}},t) = \dot{H}_{i\text{-}1}(\tilde{\bm{x}},t) + \alpha_{i\text{-}1}(H_{i\text{-}1}(\tilde{\bm{x}},t)),
\end{aligned}
\end{equation}
where $\alpha_j(t):\mathbb{R}_{\geq0}\rightarrow \mathbb{R}_{\geq0} \; \forall \, j = 0,\ldots r-1$ are class $\mathcal{K}$ and $C^{r-j}$ functions.
A safe set for every $H_i$ is defined as
\begin{equation}
\label{eq:safe sets cascade}
\begin{aligned}
\mathcal{C}_{S_i}(t) &= \{\tilde{\bm{x}} \in \mathbb{D} : H_i \geq 0 \} \;\forall i=0,\ldots r.\\
\end{aligned}
\end{equation}
\begin{definition}(modified from \cite[Def. 7]{xiao_control_2019})\label{high order barrier definition}
Let $H_{i}$ and $\mathcal{C}_{S_i}$ be defined as in equations \eqref{eq:barriers cascade} and \eqref{eq:safe sets cascade}, respectively, $\forall i=0,\ldots r$. The function $h(\tilde{\bm{x}},t) : \mathbb{D} \times \mathbb{I} \rightarrow \mathbb{R}$ is a High Order Barrier Function (HOCBF) for \eqref{eq:perturbed nonlinear system} if it is $r$ times differentiable in $\tilde{\bm{x}}$ and $t$, and there exists $\alpha_j(t):\mathbb{R}_{\geq0}\rightarrow \mathbb{R}_{\geq0} \; \forall \, j = 0,\ldots r-1$ class $\mathcal{K}$-functions such that $\alpha_j$ is of class $C^{r-j}$ and 
\begin{equation}
\label{eq:HOCBF condition}
\begin{aligned}
    & \underset{\bm{u}\in \mathbb{U}}{sup}\biggl[\pdv{H_{r\text{-}1}(\tilde{\bm{x}},t)}{t}+\mathcal{L}_fH_{r\text{-}1}(\tilde{\bm{x}},t) + \mathcal{L}_gH_{r\text{-}1}(\tilde{\bm{x}},t)\bm{u}+  \\
    & \quad \mathcal{L}_gH_{r\text{-}1}(\tilde{\bm{x}},t)\bm{d} + \alpha_{r\text{-}1}(H_{r\text{-}1}(\tilde{\bm{x}},t))\biggr]\geq0, \\ 
    & \forall \tilde{\bm{x}} \in \cap_{r\text{-}1}^{0}\mathcal{C}_{S_i}(t). 
\end{aligned}
\end{equation}
\end{definition}
If condition \eqref{eq:HOCBF condition} is satisfied, then there exists a control input $\bm{u}\in\mathbb{U}$ that renders $\cap^0_{r\text{-}1}\mathcal{C}_{S_i}(t)$ forward invariant \cite[Thm. 5]{xiao_control_2019}.
For easiness of notation we define the function $\tilde{\zeta}(\tilde{\bm{x}},\bm{u},\delta\bm{d},t):\mathbb{D}\times\mathbb{U}\times\mathbb{W}\times\mathbb{I} \rightarrow \mathbb{R}$ as
\begin{equation}\label{eq:HOCBF constraint defintion}
\begin{aligned}
\tilde{\zeta}&(\tilde{\bm{x}},\bm{u},\delta\bm{d},t):=\pdv{H_{r\text{-}1}(\tilde{\bm{x}},t)}{t}+\mathcal{L}_fH_{r\text{-}1}(\bm{x},t)+ \\& \qquad \mathcal{L}_gH_{r\text{-}1}(\tilde{\bm{x}},t)(\bm{u}+\bm{d}) + \alpha_{r\text{-}1}(H_{r\text{-}1}(\tilde{\bm{x}},t))
\end{aligned}
\end{equation}
and forward invariance of the safe set $\cap^{0}_{r\text{-}1}\mathcal{C}_{S_i}(t)$ is enforced by ensuring that the condition $\tilde{\zeta}(\tilde{\bm{x}},\bm{u},\delta\bm{d},t) \geq 0$ is met everywhere inside $\cap^0_{r\text{-}1}\mathcal{C}_{S_i}(t)$. Similarly to the function $\eta(\bm{x},\bm{u},t)$, we define the function $\zeta(\bm{x},\bm{u},t) \triangleq \tilde{\zeta}(\tilde{\bm{x}},\bm{u},\bm{0},t)$. 

%% file: content/problemStatement.tex
\label{problem-statement}
We want to maintain a set of inspector spacecrafts with dynamics as in \eqref{eq:perturbed nonlinear system}, $\epsilon$-close to a set of distinct PROs (Tab. \ref{tab:PRO table}) relative to the ISS.
The PRO set is defined to avoid collision between the inspectors as long as each inspector is within a safe corridor from its reference. Hence, no active collision avoidance is needed.
Therefore, we propose a CMPC \cite{roque_corridor_2022} control scheme subject to two HOCBF constraints for each inspector: one to constraint the maximum position tracking error and the other to constraint the maximum velocity tracking error. 
We consider the following problem.
\begin{problem}\label{main problem}
 Consider a set of $n$ inspector spacecraft $s_i\;,  i=1,\ldots n$ in relative orbit around a space vehicle with dynamics \eqref{eq:perturbed nonlinear system}, and assume the space vehicle to be in a nearly circular orbit. Consider as well $n$ distinct reference PROs $\bm{x}_{i,r}(t)=[\dr_{i,r},\dvel_{i,r}]^T$ with relative velocity $\delta \bm{v}_{r}(t)$ and position $\delta \bm{r}_{r}(t)$ according to Table \ref{tab:PRO table} such that each inspector $s_i$ is assigned to a specific reference PRO $\bm{x}_{r,i}$. We consider how to synthesize a ZOH feedback control input $K_i(\tilde{\bm{x}},t) \in \mathbb{U}$ for each inspector under dynamics \eqref{eq:perturbed nonlinear system} such that  
 $
\|\delta \tilde{\bm{r}}_{i}(t) - \delta \bm{r}_{i,r}(t)\|\leq \epsilon_{\dr} \; \forall t\in \mathbb{I}
$
with $\epsilon_{\dr}\in \mathbb{R}_{>0}$.

\end{problem}

%% file: content/sampleDataHOCBF.tex
\label{discrete time barrier function section}
To solve Problem \ref{main problem}, we define the sampled-data HOCBF, which will be applied inside the control scheme presented in Section \ref{control}. 
For a \textit{zero order hold} (ZOH) sampled-data system, the state measurements are only available at discrete time steps $k\Delta t$. Assuming that a control input $\bm{u}(k\Delta t)\in \mathbb{U}$ is applied to \eqref{eq:perturbed nonlinear system} such that $\tilde{\zeta}(\tilde{\bm{x}}(k\Delta t),\bm{u}(k\Delta t),\delta\bm{d}(k\Delta t),k\Delta t)\geq0$ is satisfied, this condition is not sufficient to guarantee safety throughout the sampling interval $[k\Delta t, (k+1)\Delta t]$. In this section, we present our first contribution to solve this problem by expanding the definition of sampled-data CBF in \cite{roque_corridor_2022} to the HOCBF case. 
\begin{lemma}\label{discrete safety roboust HOCBF}
Consider the perturbed control affine system \eqref{eq:perturbed nonlinear system} where the functions $\bm{g} :\mathbb{R}^{n} \times \mathbb{R}_{\geq 0} \rightarrow \mathbb{R}^{n \times m}$ and  $\bm{f} : \mathbb{R}^{n} \times \mathbb{R}_{\geq 0} \rightarrow \mathbb{R}^{n}$ are at least $C^{r+1}$ in  $\tilde{\bm{x}}$ and $t$ on the set $\tilde{\mathbb{X}}\times\mathbb{I}$. Let $\bm{d}(t)\in\mathbb{W}$ be a bounded unknown piece-wise differentiable disturbance defined on the compact set $\mathbb{W}$ such that $\|\bm{d}(t)\| \leq \epsilon_d \; \forall t\in \mathbb{I}$. Let $[k\Delta t,(k+1)\Delta t) \subset \mathbb{I}$ be a sampling interval for some $\Delta t >0$ such that \eqref{eq:perturbed nonlinear system} is subject to a constant bounded feedback control input $\bm{u}(t) = \bm{u}(\tilde{\bm{x}}(k\Delta t),k\Delta t)\in \mathbb{U} \; \forall t \in [k\Delta t,(k+1)\Delta t]$ shorthanded as $\bm{u}(k\Delta t)$. Furthermore consider a HOCBF $h:\mathbb{D} \times \mathbb{I} \rightarrow \mathbb{R}$ of relative degree $r$ as in Def. \ref{high order barrier definition}  that is at least 
$C^{r+1}$ on $\mathbb{D}\times\mathbb{I}$ and where $\alpha_j$ is $C^{r-j} \; \forall j=0,\ldots r-1 $. Let $\cap^0_{r\text{-}1}\mathcal{C}_{S_i}(t)\subset \mathbb{D}$ be the associated safe sets as in \eqref{eq:safe sets cascade}. Given that at time instant $k\Delta t$, $\tilde{\bm{x}}(k\Delta t) \in \cap^0_{r\text{-}1}\mathcal{C}_{S_i}(t)$, $\bm{x}(k\Delta t)\triangleq\tilde{\bm{x}}(k\Delta t)$ and that the constant feedback control input $\bm{u}(k\Delta t) \in \mathbb{U}$ satisfies 
\begin{equation}\label{eq:roboust safety condition HOBF}
\zeta(\bm{x}(k\Delta t),\bm{u}(k\Delta t),k\Delta t) - L_w\Delta t - c_w\epsilon_d \geq 0
\end{equation}
where $L_w$ is defined over the set $\mathcal{Q}\triangleq\mathbb{D}\times\mathbb{U}\times\mathbb{W}\times\mathbb{I}$ as $
L_w =\underset{(\tilde{\bm{x}},\bm{u},\delta\bm{d},t) \in \mathcal{Q}}{max}|\dot{\tilde{\zeta}}(\tilde{\bm{x}},\bm{u},\delta\bm{d},t)|,
$,
and $c_w$ is defined as $
c_w = \underset{(\tilde{\bm{x}},t)\in \mathbb{D}\times\mathbb{I}}{max}||\mathcal{L}_gH_{r\text{-}1}(\tilde{\bm{x}},t)||,
$. Then it holds that
$
\tilde{\bm{x}}(k \Delta t) \in \cap_{r\text{-}1}^{0}\mathcal{C}_{S_i}(k \Delta t) \Rightarrow  \tilde{\bm{x}}(t) \in \cap_{r\text{-}1}^{0}\mathcal{C}_{S_i}(t), \forall t \in [k\Delta t,(k+1)\Delta t]$.
\end{lemma}
\begin{proof}
Given a constant feedback input $\bm{u}(k \Delta t)$ on interval $[k\Delta t,(k+1)\Delta t]$, we assume that $\bm{u}(k \Delta t)$ and $\bm{d}$ are bounded on $[k\Delta t,(k+1)\Delta t)$ according to \eqref{eq:control and disturbances bounds}. Since $\bm{f}$ and $\bm{g}$ in \eqref{eq:perturbed nonlinear system} are $C^{r+1}$ and $\bm{d}(t)$ is piece-wise differentiable, then the solution $\tilde{\bm{x}}(t)$ to \eqref{eq:perturbed nonlinear system} is uniquely defined on an interval $[k\Delta t,\tau]\subset [k\Delta t,(k+1)\Delta t)$ for some $\tau \leq (k+1)\Delta t$ \cite[Thm. 54]{sontag_mathematical_2013}. In addition, since $\tilde{\bm{x}}(k\Delta t) \in \cap^0_{r\text{-}1}\mathcal{C}_{S_i}(t) \subset \mathbb{D}$ and by continuity of $\tilde{\bm{x}}(t)$ there exists $\tau_0\in [k\Delta t,\tau]$ such that $\tilde{\bm{x}}(t) \in \mathbb{D} \; \forall t\in [k\Delta t,\tau_0]$. The function $\tilde{\zeta}$ on $[k\Delta t,\tau_0]$ is then written as a functions of time
$\tilde{\zeta}(\tilde{\bm{x}}(t),\bm{u}(k\Delta t),\bm{d}(t),t) =\pdv{H_{r\text{-}1}(\tilde{\bm{x}}(t),t)}{t}+\mathcal{L}_fH_{r\text{-}1}(\tilde{\bm{x}}(t),t)+  \mathcal{L}_gH_{r\text{-}1}(\tilde{\bm{x}}(t),t)(\bm{u}(k\Delta t)+\bm{d}(t)) + \alpha_{r\text{-}1}(H_{r\text{-}1}(\tilde{\bm{x}}(t),t)), \forall t \in [k\Delta t,\tau_0].
$
Note that $\tilde{\zeta}(\tilde{\bm{x}}(t),\bm{u}(k\Delta t),\bm{d}(t),t)$ is a piece-wise differentiable function of time as all the functions within its definition are continuously differentiable apart from $\bm{d}(t)$, which is only piece-wise differentiable. Since this function respects the conditions in \cite[Prop 4.1.2]{scholtes_introduction_2012}, the Lipschitz constant for $\tilde{\zeta}(\tilde{\bm{x}}(t),\bm{u}(k\Delta t),\bm{d}(t),t)$ exists and is given by
$
L_w =\underset{(\tilde{\bm{x}},\bm{u},\delta\bm{d},t) \in \mathcal{Q}}{max}|\dot{\tilde{\zeta}}(\tilde{\bm{x}},\bm{u}(k\Delta t),\delta\bm{d},t)|.
$
Introducing the shorthand notation $\tilde{\zeta}_{k}(\tilde{x}(t))\triangleq\tilde{\zeta}(\tilde{\bm{x}}(t),\bm{u}(k\Delta t),\bm{d}(t),t)$ and
$\zeta_{k}(\bm{x}(t))\triangleq\zeta(\bm{x}(t),\bm{u}(k\Delta t),t)
$
we can apply the Lipschitz continuity property on $\tilde{\zeta}_k(\tilde{\bm{x}}(t))$ such that $
|\tilde{\zeta}_{k}(\tilde{\bm{x}}(t_2))-\tilde{\zeta}_{k}(\tilde{\bm{x}}(t_1))| \leq L_w |t_2-t_1| \; \forall t_1,t_2\in [k\Delta t,\tau_0]
$. The maximum negative variation of $\tilde{\zeta}_{k}(\tilde{\bm{x}}(t))$ in the interval $[k\Delta t,\tau_0]$ then satisfies
\begin{equation}\label{mid term inequality}
\begin{aligned}
\tilde{\zeta}_k(\tilde{\bm{x}}(t_2))-\tilde{\zeta}_k(\tilde{\bm{x}}(t_1)) \geq - L_w |t_2-t_1| \; \forall t_1,t_2\in [k\Delta t,\tau_0].
\end{aligned}
\end{equation}
Recalling that at time $k\Delta t$, $\tilde{\bm{x}}(k\Delta t) = \bm{x}(k\Delta t)$
the following relation between $\tilde{\zeta}_k(\tilde{\bm{x}}(k\Delta t))$ and $\zeta_k(\bm{x}(k\Delta t))$ holds 
$\tilde{\zeta}_k(\tilde{\bm{x}}(k\Delta t)) = \zeta_k(\bm{x}(k\Delta t)) + \mathcal{L}_gH_{r\text{-}1}(\tilde{\bm{x}}(k\Delta t),k\Delta t)\bm{d}(k\Delta t) \geq \zeta_k(\bm{x}(k\Delta t)) - \underset{(\tilde{\bm{x}},t)\in \mathbb{D}\times\mathbb{I}}{max}||\mathcal{L}_gH_{r\text{-}1}(\tilde{\bm{x}},t)||\,||\bm{d}|| = \zeta_k(\bm{x}(k\Delta t)) - c_w\epsilon_d$.
Replacing $t_1$ with $k\Delta t$, $t_2$ with $t$ in \eqref{mid term inequality} and by adding and subtracting $\tilde{\zeta}_k(\tilde{\bm{x}}(k\Delta t))$ on the LHS and RHS we obtain
\begin{equation}\label{eq:final zeta inequality}
\begin{aligned}
&\tilde{\zeta}_k(\tilde{\bm{x}}(t)) \geq -L_w\Delta t + \tilde{\zeta}_k(\tilde{\bm{x}}(k\Delta t)) \\&\geq -L_w\Delta t + \zeta_k(\bm{x}(k\Delta t)) - c_w\epsilon_d\;\forall t\in [k\Delta t,\tau_0].
\end{aligned}
\end{equation}
Replacing \eqref{eq:roboust safety condition HOBF} from the lemma statement in \eqref{eq:final zeta inequality} it is evident that $\tilde{\zeta}_k(\tilde{\bm{x}}(t)) \geq 0 \; \forall t\in [k\Delta t,\tau_0]$ which by \cite[Thm. 5]{xiao_control_2019} ensures that $\tilde{\bm{x}}(t)\in \cap^0_{r\text{-}1}\mathcal{C}_{S_i}(t)\; \forall t \in [k\Delta t,\tau_0]$. We will now prove that $\tilde{\bm{x}}(t) \in \cap^0_{r\text{-}1}\mathcal{C}_{S_i}(t)\; \forall t \in [k\Delta t,\tau]$ by contradiction. Suppose instead that for some $\tau_{a} \in\left(\tau_{0}, \tau\right], \tilde{\bm{x}}\left(\tau_{a}\right) \in \mathbb{D} \backslash \cap^0_{r\text{-}1}\mathcal{C}_{S_i}(t)$ and $\tilde{\bm{x}}(t) \in \mathbb{D}$ for all $t \in\left[k \Delta t,\tau_a \right]$ (i.e., the solution has left $\cap^0_{r\text{-}1}\mathcal{C}_{S_i}(t)$, but not $\left.\mathbb{D}\right)$. Then $\tilde{\bm{x}}(t)$ must leave $\cap^0_{r\text{-}1}\mathcal{C}_{S_i}(t)$ at some $t<\tau_{a}$. Furthermore, since the closed-loop dynamics are differentiable on $\mathbb{D}$, $\tilde{\bm{x}}(t)$ is uniquely defined on $\left[k \Delta t, \tau_{a}\right]$ (this is shown by repeatedly applying \cite[Thm. 54]{sontag_mathematical_2013},  since $\tilde{\bm{x}}(t)$ remains in $\mathbb{D}$ over which local differentiability of the closed-loop dynamics holds). To leave $\cap^0_{r\text{-}1}\mathcal{C}_{S_i}(t)$, $\dot{H}_{r\text{-}1}(\tilde{\bm{x}}, t)<0$ must hold on $\partial (\cap^0_{r\text{-}1}\mathcal{C}_{S_i}(t))$. The maximum negative variation of $\tilde{\zeta}_k(\tilde{\bm{x}}(t))$ is then recomputed over the interval $[k\Delta t,\tau_a]$ and is again obtained that $\tilde{\zeta}_k(\tilde{\bm{x}}(t))\geq 0 \;\forall t\in[k\Delta t, \tau_a]$ as $L_w$ and $c_w$ are independent of $\tau_a$ and $\tau_0$. Therefore we see that $\dot{H}_{r\text{-}1}(\tilde{\bm{x}}, t) \geq 0$ holds for any $\tilde{\bm{x}}(t) \in \cap^0_{r\text{-}1}\mathcal{C}_{S_i}(t), t \in\left[k \Delta t, \tau_{a}\right]$ such that $\tau_a\leq (k+1)\Delta t$. Hence, the contradiction is  reached, and so $\tilde{\bm{x}}(t)$ can never leave $\cap^0_{r\text{-}1}\mathcal{C}_{S_i}(t)$ (and thus $\mathbb{D})$ on $t \in\left[k \Delta t, \tau\right]$ with $\tau\leq (k+1)\Delta t$. Since it was showed that $\tilde{\bm{x}}(t)$ remains in a compact subset on the interval $\left[k \Delta t, \tau\right]$ (namely $\mathbb{D}$), then $\tilde{\bm{x}}(t)$ exists and is unique over the whole interval $\left[k \Delta t, (k+1)\Delta t\right]$ \cite[Prop. C.3.6]{sontag_mathematical_2013}. By the same arguments applied for the previous sub-intervals, we prove that  $\tilde{\zeta}_k(\tilde{\bm{x}}(t))\geq 0 \; \forall t\in[k\Delta t, (k+1)\Delta t]$ ensuring that $\tilde{\bm{x}}(t)\in \cap^0_{r\text{-}1}\mathcal{C}_{S_i}(t)\;\forall t \in [k\Delta t,(k+1)\Delta t]$ by \cite[Thm.5]{xiao_control_2019}.
\end{proof}
As in \cite{breeden_control_2021}, satisfying \eqref{eq:roboust safety condition HOBF} might require excessive control authority given $c_w$ and $L_w$ are calculated on $\mathcal{Q}$. For this reason, we introduce the definition of reachable set, where locally valid parameters $L_w^l$ and $c_w^l$ can be defined.
\begin{definition}\label{reachable set}
    Given a set $\mathcal{N}\subseteq\tilde{\mathbb{X}}$ we define ${}^{\Delta t}\mathcal{R}(\mathcal{N})$ as the set of all states $\tilde{\bm{x}}\in\tilde{\mathbb{X}}$ that can be reached from $\tilde{\bm{x}}\in \mathcal{N}$ in a time interval $\Delta t$ with state dynamics as in \eqref{eq:perturbed nonlinear system} and under available control input $\bm{u}\in\mathbb{U}$ according to \eqref{eq:control and disturbances bounds}.
\end{definition}
The following Lemma is proposed as a less conservative modification of Lemma \ref{discrete safety roboust HOCBF}.
\begin{lemma}\label{discrete safety roboust HOCBF local}
Consider the perturbed control affine system \eqref{eq:perturbed nonlinear system} where functions $\bm{g} :\mathbb{R}^{n} \times \mathbb{R}_{\geq 0} \rightarrow \mathbb{R}^{n \times m}$ and  $\bm{f} : \mathbb{R}^{n} \times \mathbb{R}_{\geq 0} \rightarrow \mathbb{R}^{n}$ are at least $C^{r+1}$ in  $\tilde{\bm{x}}$ and $t$ on set $\tilde{\mathbb{X}}\times\mathbb{I}$. Let $\bm{d}(t)\in\mathbb{W}$ be a bounded unknown piece-wise differentiable disturbance defined on the compact set $\mathbb{W}$ such that $\|\bm{d}(t)\| \leq \epsilon_d \; \forall t\in \mathbb{I}$. Let $[k\Delta t,(k+1)\Delta t]\subset \mathbb{I}$ be a sampling interval for some $\Delta t >0$ such that \eqref{eq:perturbed nonlinear system} is subject to a constant bounded feedback control input $\bm{u}(t) = \bm{u}(\tilde{\bm{x}}(k\Delta t),k\Delta t)\in \mathbb{U} \; \forall t \in [k\Delta t,(k+1)\Delta t]$ shorthanded as $\bm{u}(k\Delta t)$. Furthermore consider a HOCBF $h:\mathbb{D} \times \mathbb{I} \rightarrow \mathbb{R}$ of relative degree $r$ as in Def. \ref{high order barrier definition}  that is at least 
$C^{r+1}$ on $\mathbb{D}\times\mathbb{I}$ and where $\alpha_j$ is $C^{r-j} \; \forall j=0,\ldots r-1 $. 
Let $\cap^0_{r\text{-}1}\mathcal{C}_{S_i}(t)\subset \mathbb{D}$ be the safe set as in \eqref{eq:safe sets cascade}. Given that at time instant $k\Delta t$, $\tilde{\bm{x}}(k\Delta t) \in \cap^0_{r\text{-}1}\mathcal{C}_{S_i}(t)$, $\bm{x}(k\Delta t)\triangleq\tilde{\bm{x}}(k\Delta t)$ and that the constant feedback control input $\bm{u}(k\Delta t) \in \mathbb{U}$ respects the condition 
$\zeta(\bm{x}(k\Delta t),\bm{u}(k\Delta t),k\Delta t) -\quad L^l_w(\tilde{\bm{x}}(k\Delta t),\Delta t)\Delta t - c^l_w(\tilde{\bm{x}}(k\Delta t),\Delta t)\epsilon_d \geq 0$
where $L^l_w(\tilde{\bm{x}}(k\Delta t),\Delta t)$ is defined over the set ${}^{\Delta t}\mathcal{Q}(\cap^0_{r\text{-}1}\mathcal{C}_{S_i}(t))\triangleq {}^{\Delta t}\mathcal{R}(\cap^0_{r\text{-}1}\mathcal{C}_{S_i}(t))\times\mathbb{U}\times\mathbb{W}\times[k\Delta t,(k+1)\Delta t]$ as
$$
L^l_w(\tilde{\bm{x}}(k\Delta t),\Delta t) = \underset{\substack{(\tilde{\bm{x}},\bm{u},\bm{d},t) \in \\{}^{\Delta t}\mathcal{Q}(\cap^0_{r\text{-}1}\mathcal{C}_{S_i}(k\Delta t))}}{max}|\dot{\tilde{\zeta}}(\tilde{\bm{x}},\bm{u}(k \Delta t),\delta\bm{d},t)|,
$$
and the constant $c^l_w(\tilde{\bm{x}}(k\Delta t),\Delta t)$ is defined as
$$
c^l_w(\tilde{\bm{x}}(k\Delta t),\Delta t)  = \underset{{}^{\Delta t}\mathcal{R}(\cap^0_{r\text{-}1}\mathcal{C}_{S_i}(k\Delta t))\times\mathbb{I}}{max}||\mathcal{L}_g\mathcal{L}_{f}^{r-1}h(\tilde{\bm{x}},t)||.
$$
Then, $\forall t \in [k\Delta t,(k+1)\Delta t]$, it holds that
\begin{equation}\label{eq:safe set invariance in discrete time steps}
\tilde{\bm{x}}(k \Delta t) \in \cap_{r\text{-}1}^{0}\mathcal{C}_{S_i}(k \Delta t) \Rightarrow \tilde{\bm{x}}(t) \in \cap_{r\text{-}1}^{0}\mathcal{C}_{S_i}(t).
\end{equation}
\end{lemma}
\begin{proof}
The proof follows from the proof of Lemma \ref{discrete safety roboust HOCBF} by noting that inside the interval $[k\Delta t,(k+1)\Delta t]$ all the solutions $\tilde{\bm{x}}(t)$ to \eqref{eq:perturbed nonlinear system} are inside ${}^{\Delta t}\mathcal{R}(\cap^0_{r\text{-}1}\mathcal{C}_{S_i}(k\Delta t))$.
\end{proof}
\vspace{-2mm}
Let us now define a sampled-data HOCBF (SD-HOCBF).
\begin{definition}[SD-HOCBF]
\label{def:SDHOCBF}
Consider a HOCBF $h: \mathbb{D} \times \mathbb{I} \rightarrow \mathbb{R}$ as in Definition \ref{high order barrier definition} with relative degree $r$ with respect to \eqref{eq:perturbed nonlinear system} and corresponding safe set $\cap^0_{r\text{-}1}\mathcal{C}_{S_i}(t)$ according to \eqref{eq:safe sets cascade}. The function $h$ is a sampled-data HOCBF (SD-HOCBF) for a given $\Delta t>0$  if for any point $\tilde{\bm{x}} \in \cap^0_{r\text{-}1}\mathcal{C}_{S_i}(k\Delta t)$ and $k \in \mathbb{N}_{0}$, there is a constant feedback input $\bm{u}(\bm{x}(k \Delta t), k \Delta t) \in \mathbb{U}$, shortened to $\bm{u}(k \Delta t)$, where $\bm{x}(k\Delta t)\triangleq\tilde{\bm{x}}(k\Delta t)$ , such that
\begin{equation}\label{eq:sdhobf main definition equation}
\begin{aligned}
&\underset{\bm{u}\in \mathbb{U}}{sup}[\zeta(\bm{x}(k \Delta t),\bm{u}(\Delta t),k\Delta t)  -\\& \quad L^l_w(\tilde{\bm{x}}(k\Delta t),\Delta t) \Delta t - c^l_w(\tilde{\bm{x}}(k\Delta t),\Delta t)\epsilon_d] \geq 0.
\end{aligned}
\end{equation}
\end{definition}

%% file: content/control.tex
\label{control}
Here we present the two SD-HOCBF that will be used in CMPC scheme applied to solve Problem \ref{main problem}, 
 \begin{subequations}\label{eq:the two barrier constarints}
 \begin{equation}\label{eq:position barrier function}
h_{\dr}(\tilde{\bm{x}},t) = \epsilon^2_{\dr} -  \|\delta \tilde{\bm{r}} - \delta \bm{r}_r(t)\|^2 = \epsilon^2_{\dr} -  \|\tilde{\bm{e}}_{\dr}(t)\|^2,
\end{equation}
\begin{equation}\label{eq:velocity barrier function}
h_{\dvel}(\tilde{\bm{x}},t) = \epsilon^2_{\dvel} -  \|\delta \tilde{\bm{v}} - \delta \bm{v}_r(t)\|^2 = \epsilon^2_{\dvel} -  \|\tilde{\bm{e}}_{\dvel}(t)\|^2
\end{equation}
\end{subequations}
where $h_{\dr}(\tilde{\bm{x}},t)$ bounds the position error to tackle Problem \ref{main problem}, and $h_{\dvel}(\tilde{\bm{x}},t)$ bounds the maximum velocity error to provide recursive feasibility guarantees on the satisfaction of the former, as we will demonstrate.\par
It can be shown that $h_{\dr}(\tilde{\bm{x}},t)$ is a relative degree two HOCBF for \eqref{eq:perturbed nonlinear system} while $h_{\dvel}(\tilde{\bm{x}},t)$ is relative degree one. Given the class-$\mathcal{K}$ functions $\alpha_{\dr_0}(x) = p_{\dr_0}x$, $\alpha_{\dr_1}(x) = p_{\dr_1}x$ for $h_{\dr}$ and $\alpha_{\dvel_0}(x) = p_{\dvel_0}x$ for $h_{\dvel}$, with $p_{\dr_0},p_{\dr_1},p_{\dvel_0}\in \mathbb{R}_{>0}$, we define functions $\tilde{\zeta}_{\dr}$ and $\tilde{\zeta}_{\dvel}$ as
\begin{subequations}\label{eq:zeta def}
\begin{equation}\label{eq:zeta dvel tilde}
\tilde{\zeta}_{\dvel}=-2\tilde{\bm{e}}^T_{\dvel}(\bm{f}_v+ \bm{u} + \bm{d}) + p_{\dvel±_0}(\epsilon_{\dvel}^2 - \|\tilde{\bm{e}}_{\dvel}\|^2),
\end{equation}
\begin{equation}\label{eq:position zeta}
\begin{aligned}
&\tilde{\zeta}_{\dr} =-2\|\tilde{\bm{e}}_{\dvel}\|^2 - 2\tilde{\bm{e}}_{\dr}^T(\bm{f}_v + \bm{d}+\bm{u})-\\ &  2(p_{\dr_0}+p_{\dr_1})(\tilde{\bm{e}}_{\dr}^T\tilde{\bm{e}}_{\dvel})+ p_{\dr_0}p_{\dr_1}(\epsilon_{\dr}^2 - \|\tilde{\bm{e}}_{\dr}\|^2).
\end{aligned}
\end{equation}
\end{subequations}
The safe set definitions are then given by
$
\mathcal{C}_{S,\dr}(t) = \{\tilde{\bm{x}}\in\tilde{\mathbb{X}} : H_{\dr_0}(\tilde{\bm{x}},t) \geq 0 \land  H_{\dr_1}(\tilde{\bm{x}},t) \geq 0 \}, 
\mathcal{C}_{S,\dvel}(t) = \{\tilde{\bm{x}}\in\tilde{\mathbb{X}} : H_{\dvel_0}(\tilde{\bm{x}},t) \geq 0 \}, 
\mathcal{C}_{S,x}(t) = \mathcal{C}_{S,\dvel}(t) \cup  \mathcal{C}_{S,\dr}(t)\}$, with $H_{\dr_1}$,$H_{\dr_0}$,$H_{\dvel_0}$ defined in \eqref{eq:barriers cascade}. 
It is evident from \eqref{eq:position zeta} that $\|\tilde{\bm{e}}_{\dvel}\|$ may grow unbounded inside the safe set if not bounded by $h_{\dvel}(\tilde{\bm{x}},t)$. As $\tilde{\zeta}_{\dr}$ (and $\zeta_{\dr}$) depends directly on $\|\tilde{\bm{e}}_{\dvel}\|$, it is crucial to ensure boundedness of such term to analyze the satisfaction of \eqref{eq:sdhobf main definition equation}. In the remainder, we will refer to $L^l_w(\tilde{\bm{x}},t)$ and $c^l_w(\tilde{\bm{x}},t)$ from Lemma \ref{discrete safety roboust HOCBF local} as computed for $h_{\dr}$ and $h_{\dvel}$ with $L^l_{\dr}$,$c^l_{\dr}$ and $L^l_{\dvel}$,$c^l_{\dvel}$ respectively.\par 
Computing $L^l_w(\tilde{\bm{x}},\Delta t)$ and $c^l_w(\tilde{\bm{x}},\Delta t)$ requires the online computation of ${}^{\Delta t}\mathcal{R}(\mathcal{C}_{S,x}(t))$. 
While an analytical definition is impractical, we propose to find a set ${}^{\Delta t }\bar{\mathcal{R}}(\mathcal{C}_{S,x}(t)) \supset {}^{\Delta t }\mathcal{R}(\mathcal{C}_{S,x}(t))$ such that suitable constants $\bar{L}^l_w(\Delta t)\geq L^l_w(\tilde{\bm{x}},\Delta t)$, $\bar{c}^l_w(\Delta t)\geq c^l_w(\tilde{\bm{x}},\Delta t)$ can be computed only based on the current sampling time $\Delta t$. This result is achieved in two steps: i) we first derive a constant $\bar{a}$ that upper bounds the maximum total acceleration that \eqref{eq:perturbed nonlinear system} is subject to; and ii) based on $\bar{a}$ we compute the maximum position and velocity tracking errors that can be reached by \eqref{eq:perturbed nonlinear system} in a sampling interval $[k\Delta t,(k+1)\Delta t]$. Namely, we compute the maximum acceleration that \eqref{eq:perturbed nonlinear system} is subject  as
\begin{equation}\label{eq:max acc}
 \underset{(\bm{x},t)\in \tilde{\mathbb{X}}\times \mathbb{I}}{max}\|\delta \dot{\tilde{\bm{v}}}\|=\|\bm{f}_{v}+ \bm{u} + \bm{d}\| \leq \epsilon_f + \epsilon_u + \epsilon_d = \bar{a}   
\end{equation}
with $\epsilon_f = \underset{(\tilde{\bm{x}},t)\in \tilde{\mathbb{X}}\times \mathbb{I}}{max}\|\bm{f}_v(\tilde{\bm{x}},t)\|$. Note that in \eqref{eq:max acc} the function $\bm{g}$ was omitted as it is the identity matrix. Next,  the maximum velocity and position tracking errors are computed assuming that \eqref{eq:perturbed nonlinear system} is forced under constant acceleration $\bar{a}$ during the sampling interval $[k\Delta t,(k+1)\Delta t]$. Under these assumptions, we compute the maximum velocity and position tracking errors in first-order approximation for the time interval $[k\Delta t,(k+1)\Delta t]$ as $\|\tilde{\bm{e}}_{\dvel}(t) - \tilde{\bm{e}}_{\dvel}(k\Delta t)\| \leq 
\|\delta \tilde{\bm{v}}(t) - \delta \tilde{\bm{v}}(k\Delta t)\| + \|\delta \bm{v}_r(t) - \delta \bm{v}_r(k\Delta t)\| \leq \bar{a}\Delta t + \bar{a}_r\Delta t \;\forall t\in [k\Delta t,(k+1)\Delta t]$
and 
$\|\tilde{\bm{e}}_{\dr}(t) - \tilde{\bm{e}}_{\dr}(k\Delta t)\| \leq \|\delta \tilde{\bm{r}}(t) - \delta \tilde{\bm{r}}(k\Delta t)\| + \|\delta \bm{r}_r(t) - \delta \bm{r}_r(k\Delta t)\| \leq \hspace{0.2cm} \bar{a}\frac{\Delta t^2}{2} + \epsilon_{\dvel}\Delta t + \bar{a}_r\frac{\Delta t^2}{2} + \bar{v}_{r}\Delta t\;\forall t\in [k\Delta t,(k+1)\Delta t]
$; where $\bar{v}_r\triangleq\underset{t\in\mathbb{I}}{max}\|\delta \bm{v}_{r}(t)\|$ and $\bar{a}_r\triangleq\underset{t\in\mathbb{I}}{max}\|\delta \bm{a}_{r}(t)\|$ (Tab. \ref{tab:PRO table}). Note that we applied the formulas for linear displacement under constant acceleration to compute the maximum velocity and position error bounds. Moreover, the constants $\bar{a}$ and $\bar{v}$ exist as a PRO is a periodic trajectory. For completeness, we define $\bar{r}_r\triangleq\underset{t\in\mathbb{I}}{max}\|\delta \bm{r}_{r}(t)\|$. If we define two new functions $\gamma_{\dvel}=\bar{\epsilon}^2_{\dvel}-\|\tilde{\bm{e}}_{\dvel}\|^2$ and $\gamma_{\dr} = \bar{\epsilon}^2_{\dr}-\|\tilde{\bm{e}}_{\dr}\|^2$ with
$\bar{\epsilon}_{\dvel}(\Delta t) \triangleq \epsilon_{\dvel} + \bar{a}\Delta t + \bar{a}_r \Delta t$ and $\bar{\epsilon}_{\dr}(\Delta t) \triangleq \epsilon_{\dr} + \bar{a} \frac{\Delta t^2}{2} +\epsilon_{\dvel}\Delta t + \bar{a}_r \frac{\Delta t^2}{2} +\bar{v}_r\Delta t$; we can then define $
{}^{\Delta t }\bar{\mathcal{R}}(\mathcal{C}_{S,x}(t))$ as $
{}^{\Delta t }\bar{\mathcal{R}}(\mathcal{C}_{S,x}(t)) = \{ \tilde{\bm{x}}\in\tilde{\mathbb{X}} : \gamma_{\dvel}(\tilde{\bm{x}},t)\geq 0 \land \gamma_{\dr}(\tilde{\bm{x}},t)\geq 0\}
$. 
We next show how to compute valid constants $\bar{L}_w^l(\Delta t)\geq L_w^l(\tilde{\bm{x}},\Delta t)$ and $\bar{c}_w^l(\Delta t)\geq c_w^l(\tilde{\bm{x}},\Delta t)$ over ${}^{\Delta t}\bar{\mathcal{R}}(\mathcal{C}_{S,x}(t))$ such that the conditions for Lemma \ref{discrete safety roboust HOCBF local} are still satisfied by $\bar{L}_w$ and $\bar{c}_w$. The time derivative of $\tilde{\zeta}_{\dr}$ and  $\tilde{\zeta}_{\dvel}$ is derived as $
\dot{\tilde{\zeta}}_{\dvel}= -2\tilde{\bm{e}}_{\dvel}^T\left(\dot{\bm{f}}_v+\dot{\bm{d}}\right)-2\|\bm{f}_v + \bm{u} + \bm{d}\|^2-2p_{\dvel_0}\tilde{\bm{e}}_v^T(\bm{f}_v + \bm{u} + \bm{d})
$
and 
$
\dot{\tilde{\zeta}}_{\dr} =
 - 6\tilde{\bm{e}}_{\dvel}^T(\bm{f}_v +\bm{u} + \bm{d} ) - 2\tilde{\bm{e}}_{\dr}^T(\dot{\bm{f}_v} + \dot{\bm{d}})-\\2(p_{\dr_0}+p_{\dr_1})(\tilde{\bm{e}}_{\dvel}^T\tilde{\bm{e}}_{\dvel} + \tilde{\bm{e}}_{\dr}^T(\bm{f}_v +\bm{u} + \bm{d}  ))- 
 p_{\dr_0}p_{\dr_1}(2\tilde{\bm{e}}_{\dvel}^T\tilde{\bm{e}}_{\dr})$. Noting that  $\mathcal{L}_{g}H_{\dvel_0}(\tilde{\bm{x}},t) =-2\tilde{\bm{e}}_{\dvel}$ and $\mathcal{L}_{g}H_{\dr_1}(\tilde{\bm{x}},t) =-2\tilde{\bm{e}}_{\dr}$; we can obtain $\bar{L}_{\dr}^{l}(\Delta t)$,$\bar{L}_{\dvel}^{l}(\Delta t)$,$\bar{c}_{\dr}^{l}(\Delta t)$,$\bar{c}_{\dvel}^{l}(\Delta t)$ by applying the inequality $\bm{a}^T\bm{b}\leq\|\bm{a}\|\|\bm{b}\|$ in $\dot{\tilde{\zeta}}_{\dr}$ and $\dot{\tilde{\zeta}}_{\dvel}$, which yields
$
\bar{L}^l_{\delta v}(\Delta t) = 2\epsilon_{\dvel}\beta +  2\bar{a}^2 +  2p_{\dvel_0}\bar{\epsilon}_{\dvel}\bar{a} \geq L^l_{\delta v}(\tilde{\bm{x}},\Delta t) 
$,
$
\bar{c}^l_{\delta v}(\Delta t) = 2\bar{\epsilon}_{\dvel}\geq c^l_{\delta v}(\tilde{\bm{x}},\Delta t) 
$,
$
\bar{L}^{l}_{\dr}(\Delta t) = 6\bar{a} + 2(p_{\dr_0}+p_{\dr_1})(\bar{\epsilon}^2_{\dvel} + \bar{\epsilon}_{\dr}\bar{a})+2\epsilon_{\dr}\beta + 2p_{\dr_0}p_{\dr_1}(\bar{\epsilon}_{\dr}\bar{\epsilon}_{\dvel})\geq L^l_{\dr}(\tilde{\bm{x}},\Delta t)
$,
$
\bar{c}^{l}_{\dvel}(\Delta t) = 2\bar{\epsilon}_{\dr} \geq c^l_{\dvel}(\tilde{\bm{x}},\Delta t), \forall \tilde{\bm{x}} \in \mathcal{C}_{S,x}(t)
$,
where $\beta = \underset{{}^{\Delta t}\mathcal{Q}(\cap^0_{r\text{-}1}\mathcal{C}_{S_i}(t))}{max}\|\dot{\bm{f}}_v\| + \underset{\tilde{{\mathbb{X}}}\times \mathbb{I}}{max}\|\dot{\bm{d}}\|$.

\subsection{Corridor MPC}
Now that  $\bar{L}^{l}_{\dr}(\Delta t)$,$\bar{L}^{l}_{\dvel}(\Delta t)$,$\bar{c}^{l}_{\dr}(\Delta t)$,$\bar{c}^{l}_{\dvel}(\Delta t)$ are defined, the CMPC scheme \cite[Eqs. 17-18]{roque_corridor_2022} applied to control each inspector along its trajectory is given as
\begin{subequations}\label{eq:CMPC}
\begin{equation}
J_N^*(\tilde{\bm{e}}_x(k \Delta t))=\min _{\mathbf{u}_k^*} J_N(\bm{e}_x, \bm{u}) 
\end{equation}
\begin{equation}\label{eq:dyn const}
\begin{aligned}
\bm{x}((m+1)|k \Delta t)=\eta_d(\bm{x}(m|k\Delta t), \bm{u}(m|k\Delta t),k\Delta t) 
\end{aligned}
\end{equation}
\begin{equation}\label{eq:zeta dvel const}
\zeta_{\dvel}(\bm{x}(0|k\Delta t),\bm{u}(0|k\Delta t),0) - \bar{L}^{l}_{\dvel}\Delta t-\bar{c}^l_{\dvel}\epsilon_d\geq 0
\end{equation}
\begin{equation}\label{eq:zeta dr const}
\zeta_{\dr}(\bm{x}(0|k\Delta t),\bm{u}(0|k\Delta t),0) - \bar{L}^{l}_{\dr}\Delta t-\bar{c}^l_{\dr}\epsilon_d\geq 0
\end{equation}
\begin{equation}\label{eq:control const}
\bm{u}(m|k \Delta t) \in \mathbb{U}, \quad \forall m \in \mathbb{N}_{[0, N-1]} 
\end{equation}
\begin{equation}
\bm{e}_x(n|k \Delta t)=\bm{x}(n|k\Delta t)-\bm{x}_r(n| k \Delta t) 
\end{equation}
\begin{equation}\label{eq:state measurement}
 \bm{x}(0| k \Delta t)=\tilde{\bm{x}}(k \Delta t), \quad \forall n \in \mathbb{N}_{[0, N]}
\end{equation}
\end{subequations}
where
\begin{equation}\label{eq:cost function}
\begin{aligned}
J_N(\bm{e}_x, \bm{u})= & \sum_{n=0}^{N-1} \|\bm{e}_x(n|k \Delta t)\|_{Q} + \|\bm{u}(n| k \Delta t)\|_{R} \\
& +V(\bm{e}_x(N|k \Delta t)),
\end{aligned}
\end{equation}
and where $V(\bm{e}_x(N|k \Delta t))$ is a positive definite function of the state error. The solution to the CMPC is the optimal solution of the finite horizon optimal control problem (FHOC) in \eqref{eq:CMPC}. Such solution is an optimal control trajectory $\bm{u}_k^*= [\bm{u}^*(0|k\Delta t),\ldots \bm{u}^*((N-1)|k\Delta t)]$ and a corresponding optimal state trajectory $\bm{x}^*= [\bm{x}^*(0|k\Delta t),\ldots \bm{x}^*_k((N-1)|k\Delta t)]$ so that the cost function $J_N(\bm{e}, \bm{u})$ is minimised along the $N$-steps receding horizon of the FHOC problem. In \eqref{eq:CMPC}, the constraint \eqref{eq:dyn const} forces the state to evolve according to the nominal discrete dynamics of each inspector, \eqref{eq:zeta dvel const} and  \eqref{eq:zeta dr const}  are the SD-HOCBF constraint on the velocity and position respectively, \eqref{eq:control const} is the control constraint and \eqref{eq:state measurement} sets the initial state of the state trajectory $\bm{x}^*$ to be equal to the measured state $\tilde{\bm{x}}(k\Delta t)$. Once  the solution $\bm{u}_k^*$ to \eqref{eq:CMPC} is found, the feedback control $K(\tilde{\bm{x}},k\Delta t)=\bm{u}^*(0|k\Delta t)$ is applied in a ZOH fashion on \eqref{eq:perturbed nonlinear system} during the interval $[k\Delta t,(k+1)\Delta t]$. Since at $k\Delta t$ we have $\tilde{\bm{x}}\in \mathcal{C}_{S,x}(k\Delta t)$, Lemma \ref{discrete safety roboust HOCBF local} guarantees that $\tilde{\bm{x}}(t)\in \mathcal{C}_{S,x}(t)\,\forall t\in [k\Delta t,(k+1)\Delta t]$ under $K(\tilde{\bm{x}},k\Delta t)$.

Given the maximum control input $\epsilon_u$ and a sampling interval $\Delta t$, the recursive feasibility of the given CMPC scheme can be assessed numerically by applying the feasibility checking algorithm by \cite{tan_compatibility_2022}. Namely, we verify the compatibility of constraints  \eqref{eq:zeta dvel const}-\eqref{eq:control const} by solving the following quadratic program (QP) \begin{subequations}\label{eq:feasibility original}
\begin{equation}\label{eq:minimization}
    \underset{\bm{q},\bm{u}}{min} -\bm{1}^T\bm{q}
\end{equation}
\begin{equation}\label{eq:positive t}
    \bm{I}_{2}\bm{q} \geq \bm{0}\;, \|\bm{u}\|^2 \leq \epsilon_{u}^2 \;
\end{equation}
\begin{equation}\label{eq:position sd-hocbf}
\begin{gathered}
    \zeta_{\dr}(\bm{x},t,\bm{u})  \geq \bar{L}^l_{\dr}\Delta t + \bar{c}^l_{\dr}\epsilon_d+ q_1
\end{gathered}
\end{equation}
\begin{equation}
\begin{gathered}
\zeta_{\dvel}(\bm{x},t,\bm{u})   \geq\bar{L}^l_{\dvel}\Delta t + \bar{c}^l_{\dvel}\epsilon_d + q_2,\label{eq:velocity sd-hocbf}
\end{gathered}
\end{equation}
\end{subequations}
over a dense discretization of $\bm{x}$ and $t$, yielding a seven-dimensional parameter space for problem \eqref{eq:feasibility original}. Note that $\bm{q}=[q_1,q_2]^T$ is only a slack variable that defines how robustly \eqref{eq:position sd-hocbf}-\eqref{eq:velocity sd-hocbf} can be satisfied. We propose to lower the dimensionality of the problem by offering a more conservative solution, based on the fact that $\bm{f}_{v}(\tilde{\bm{x}},t) \leq \epsilon_f \forall (\tilde{\bm{x}},t)\in \tilde{\mathbb{X}}\times\mathbb{I}$. In particular, we 
note that both $\zeta_{\dr}$ and  $\zeta_{\dvel}$ are directly functions of the $\bm{x}$ and $\bm{t}$ through $\bm{f}_{v}(\bm{x},t)$. By upper bounding $\bm{f}_{v}(\bm{x},t)$ with $\epsilon_f$, we can assess the feasibility of the CMPC scheme by solving a new QP parameterised over $\bm{e}_{\dvel}$ and $\bm{e}_{\dr}$ instead of $\bm{x}$ and $t$ \cite[Eq. 11]{tan_compatibility_2022}
\begin{subequations}\label{eq:feasibility grid}
\begin{equation}\label{eq:minimization approx}
    \underset{\bm{q},\bm{u}}{min} -\bm{1}^T\bm{q}
\end{equation}
\begin{equation}\label{eq:positive t approx}
    \bm{I}_{2}\bm{q} \geq \bm{0}\;, \|\bm{u}\|^2 \leq \epsilon_{u}^2 \;
\end{equation}
\begin{equation}\label{eq:modified position sd-hocbf}
\begin{gathered}
    -2\|\tilde{\bm{e}}_{\dvel}\|^2 - 2\|\tilde{\bm{e}}_{\dr}\|\epsilon_f + 2\tilde{\bm{e}}_{\dr}^T\bm{u}-\\2(p_{\dr 0}+p_{\dr 1})(\tilde{\bm{e}}_{\dr}^T\tilde{\bm{e}}_{\dvel})+p_{\dr0}p_{\dr 1}(\epsilon_{\dr}^2 - \|\tilde{\bm{e}}_{\dr}\|^2) \\   \geq \bar{L}^l_{\dr}\Delta t + \bar{c}^l_{\dr}\epsilon_d+ q_1
\end{gathered}
\end{equation}
\begin{equation}
\begin{gathered}
-2\|\tilde{\bm{e}}_{\dvel}\|\epsilon_f -2\tilde{\bm{e}}_{\dvel}^T \bm{u} + p_{\dvel 0}(\epsilon_{\dvel}^2 - \|\tilde{\bm{e}}_{\dvel}\|^2)\\   \geq\bar{L}^l_{\dvel}\Delta t + \bar{c}^l_{\dvel}\epsilon_d + q_2.\label{eq:modified velocity sd-hocbf}
\end{gathered}
\end{equation}
\end{subequations}
where constraints \eqref{eq:modified position sd-hocbf} and \eqref{eq:modified velocity sd-hocbf} are a modified version of the SD-HOCBF constraints in \eqref{eq:zeta dvel const},\eqref{eq:zeta dr const}. More specifically, the terms
$-\tilde{\bm{e}}^T_{\dr}\bm{f}_{v}$ and $-\tilde{\bm{e}}^T_{\dvel}\bm{f}_{v}$ as $-\|\tilde{\bm{e}}_{\dr}\|\epsilon_f$ and $-\|\tilde{\bm{e}}_{\dvel}\|\epsilon_f$ are lower bounded, such that $\zeta_{\dvel}$ and $\zeta_{\dr}$ are maximally negatively decreased by the dynamic acceleration $\bm{f}_{\dvel}$. With such formulation, the satisfaction of \eqref{eq:feasibility grid} only depends on i) the position tracking error magnitude $\|\bm{e}_{\dr}\|$; ii) the velocity tracking error magnitude $\|\bm{e}_{\dvel}\|$; and iii) the angle $\alpha =\arccos(\frac{\bm{e}_{\dvel}\cdot \bm{e}_{\dvel}}{\|\bm{e}_{\dr}\|\|\bm{e}_{\dvel}\|})$. We can then iteratively check for infeasibility of \eqref{eq:feasibility grid} over a dense discretization of $[0,\epsilon_{\dr}]\times[0,\epsilon_{\dvel}]\times [0,\pi]$, which is only a three-dimensional grid instead of seven-dimensional one. This optimization problem can be solved offline during the mission design process. Both $\epsilon_{u}$ and $\Delta t$ are typically constrained by the specific hardware at hand, but they could also be treated as design parameters by increasing the dimensionality of the feasibility test by two. In addition, the design choice of the parameter $\epsilon_{\delta v}$  might need to be reiterated during this process if the feasibility of \eqref{eq:feasibility grid} is not achieved.

%% file: content/results.tex
\label{results}
We simulate three inspectors as they orbit the ISS to accomplish the inspection mission (Fig. \ref{fig:PRO visualization}). We leverage the optimization library \textit{CasADi} to solve the nonlinear optimal control scheme \eqref{eq:CMPC} and the QP in \eqref{eq:feasibility grid} offline.
\begin{table}[h]
\vspace{-2mm}
    \centering
    \begin{tabular}{cllll}
    \toprule
    Parameter & Inspector 1 & Inspector 2 & Inspector 3 & units\\
    \midrule
$\bar{L}^l_{\dvel}$  & 1.315$\times 10^{-3} $& 1.410$\times 10^{-3} $& 1.568$\times 10^{-3}$  & \unit{-} \\ 
$\bar{L}^l_{\dr}$    & 5.036$\times 10^{-2} $& 5.414$\times 10^{-2} $& 6.038$\times 10^{-2} $&\unit{-} \\ 
$\bar{c}^l_{\dvel}$  & 2.702$\times 10^{-1} $& 2.703$\times 10^{-1} $& 2.704$\times 10^{-1} $& \unit{-}\\ 
$\bar{c}^l_{\dr}$    & 1.405$\times 10^{1} $& 1.406$\times 10^{1} $& 1.407$\times 10^{1} $& \unit{-}\\ 
$\Delta t$      & 1.000$\times 10^{-1} $& 1.000$\times 10^{-1} $& 1.000$\times 10^{-1} $& \unit{s}\\ 
$\epsilon_f$     & 8.872$\times 10^{-4} $& 1.254$\times 10^{-3} $& 1.860$\times 10^{-3}$& \unit{\meter\second^{-2}}\\ 
$\epsilon_{\dr}$    & 7.000$\times 10^{0} $& 7.000$\times 10^{0} $& 7.000$\times 10^{0}$ & \unit{\meter}\\ 
$\epsilon_{\dvel}$    & 1.330$\times 10^{-1} $& 1.330$\times 10^{-1} $& 1.330$\times 10^{-1} $ &\unit{\meter\second^{-1}}\\ 
$\bar{\epsilon}_{\dr}$  & 7.025$\times 10^{0} $& 7.029$\times 10^{0} $& 7.037$\times 10^{0}$   &\unit{\meter} \\ 
$\bar{\epsilon}_{\dvel}$  & 1.351$\times 10^{-1} $& 1.351$\times 10^{-1} $& 1.352$\times 10^{-1} $&\unit{\meter\second^{-1}}\\ 
$\epsilon_{d}$    & 1.577$\times 10^{-6} $& 2.205$\times 10^{-6} $& 3.243$\times 10^{-6}$         &\unit{\meter\second^{-2}}\\ 
$p_{\delta r_0}$         & 2.000$\times 10^{-2} $& 2.000$\times 10^{-2} $& 2.000$\times 10^{-2}$ &\unit{-}\\ 
$p_{\delta r_1}$        & 5.000$\times 10^{-2} $& 5.000$\times 10^{-2} $& 5.000$\times 10^{-2}$   &\unit{-}\\ 
$p_{\delta v_0}$         & 5.000$\times 10^{-2} $& 5.000$\times 10^{-2} $& 5.000$\times 10^{-2}$  &\unit{-}\\ 
$\bar{a}$         & 2.089$\times 10^{-2} $& 2.126$\times 10^{-2} $& 2.186$\times 10^{-2} $& \unit{\meter\second^{-2}}\\  
$\bar{a}_r$  & 1.266 $\times 10^{-4}$& 1.789$\times 10^{-4} $& 2.653 $\times 10^{-4}$      & \unit{\meter\second^{-2}}\\ 
$\bar{v}_r$  & 1.125$\times 10^{-1} $& 1.590$\times 10^{-1} $& 2.358$\times 10^{-1}$      & \unit{\meter\second^{-1}}\\  
$\epsilon_u$        & 2.000$\times 10^{-2} $& 2.000$\times 10^{-2} $& 2.000$\times 10^{-2}$& \unit{\meter\second^{-2}}\\  
$\beta$        & 6.023$\times 10^{-4} $& 8.236$\times 10^{-4} $& 1.189$\times 10^{-3} $& \unit{\meter\second^{-3}}\\  
$\bar{r}_r$ & 1.000$\times 10^{1} $& 1.414$\times 10^{1} $& 2.096$\times 10^{1} $&\unit{\meter}\\ 
$\alpha_w$ & 0.000$\times 10^{0} $& 0.000$\times 10^{0} $& 0.000$\times 10^{0} $&\unit{\radian}\\ 
$\alpha_r$ & 1.571$\times 10^{0} $& 1.571$\times 10^{0} $& 1.571$\times 10^{0} $&\unit{\radian}\\ 
$\rho_r$ & 5.000$\times 10^{1} $& 6.400$\times 10^{1} $& 7.800$\times 10^{1} $&\unit{\meter}\\ 
$\rho_s$ & 0.000$\times 10^{0} $& 0.000$\times 10^{0} $& 0.000$\times 10^{0} $&\unit{\meter}\\ 
$\rho_w$ & 0.000$\times 10^{0} $& 6.000$\times 10^{1} $& 1.400$\times 10^{1} $&\unit{\meter}\\ 
\bottomrule
    \end{tabular}
    \caption{Parameters value for the three different agents.}
    \label{tab:CMPC parameters table}
    \vspace{-3mm}
\end{table}

The simulations were performed on a 2.3 GHz Dual-Core Intel Core i5 CPU with 8GB of RAM. 
The simulation considers zonal harmonic terms up to order six and an exponential atmospheric model.
The inspectors are 6U-CubeSats with a mass of $10$ \unit{\kilo\gram} and an omnidirectional propulsion system delivering a maximum thrust of $200$ \unit{\milli\newton} (and hence 0.02 \unit{\meter\second^{-2}} in acceleration).
Similar specifications were applied in \cite{nakka_information-based_2021} giving a realistic scenario for the inspection mission.
The ISS has a mass of $419400$ \unit{\kilo\gram} and no actuation capabilities. 
The reference ISS orbit parameters are obtained from New Horizon Database on 2023-Feb-04 00h:00m:00s. 

The inspectors are deployed on three separate PRO trajectories (Fig.~\ref{fig:PRO visualization}) according to the parameters in Tab \ref{tab:CMPC parameters table}, with a non-zero initial tracking error contained inside the initial safe set.
Namely, the initial condition at time $t_0=0$ for each inspector is $\tilde{\bm{x}}_1(0)=[55.70,1.08,2.43,1.73\times10^{-2},-9.23\times10^{-2},8.00\times10^{-3}]$, $\tilde{\bm{x}}_2(0)=[ 67.72,3.27,3.88,-2.5\times10^{-3},-1.36\times10^{-1},7.01\times10^{-2}]$ and $\tilde{\bm{x}}_3(0)=[ 82.63,0.65,4.21,-1.39\times10^{-2}  ,-4.85\times10^{-2},1.61\times10^{-1}]$.
Moreover, we have $\mathbb{U}=\{\bm{u}\in\mathbb{R}^3: \|\bm{u}\|\leq 0.02\}$  and $\mathbb{I}= [0,T]$ with $T$ being the period of the assigned PRO (which is approximately 93 \unit{\minute}). The set $\tilde{\mathbb{X}}$ is chosen for each inspector $i=1,2,3$ as $\tilde{\mathbb{X}}_i=\{\tilde{\bm{x}}\in\mathbb{R}^6: \|\delta \bm{r}\|\leq k_1\bar{r}_{r,i} \land \|\delta \bm{v}\|\leq k_2\bar{v}_{r,i}\}$ with $k_1,k_2\in \mathbb{R}_{\geq0}$ being two sufficiently large constants that realistically capture the workspace.
For this mission simulation, we used $k_1=k_2=1.4$, and Table.~\ref{tab:CMPC parameters table} summarises the CMPC parameters.
We considered $Q=diag([50\bm{I}_{3\times3},59.17\bm{I}_{3\times3}])$, $R=50\bm{I}_{3\times3}$ and $V(\bm{e}_x)=2\bm{e}^T_xP\bm{e}_x$ where P is obtained by solving the discrete algebraic Riccati equation with a linearised dynamics along the reference trajectory. 
The time step was chosen as $\Delta t=0.1$\unit{s} with $N=25$ horizon steps. 
The solution to \eqref{eq:feasibility grid} took 105 \textit{min}.
The simulation results are illustrated in Fig.~\ref{fig:final barriers result} for a time interval of 3 min. 
We observe that each inspector can simultaneously satisfy \eqref{eq:position sd-hocbf}-\eqref{eq:velocity sd-hocbf} such that the safety specification in Problem \ref{main problem} is respected.
Particularly, we notice how the optimal control strategy results in an accelerate-coast-break profile: the velocity error is first increased to its maximum norm to lower the position tracking error as fast as possible, and then the inspector is left under minimum actuation until the velocity and position errors are driven to zero.
\begin{figure}
    \centering
\includegraphics[width=0.48\textwidth,trim={0cm 3.2cm 0 3.2cm},clip]{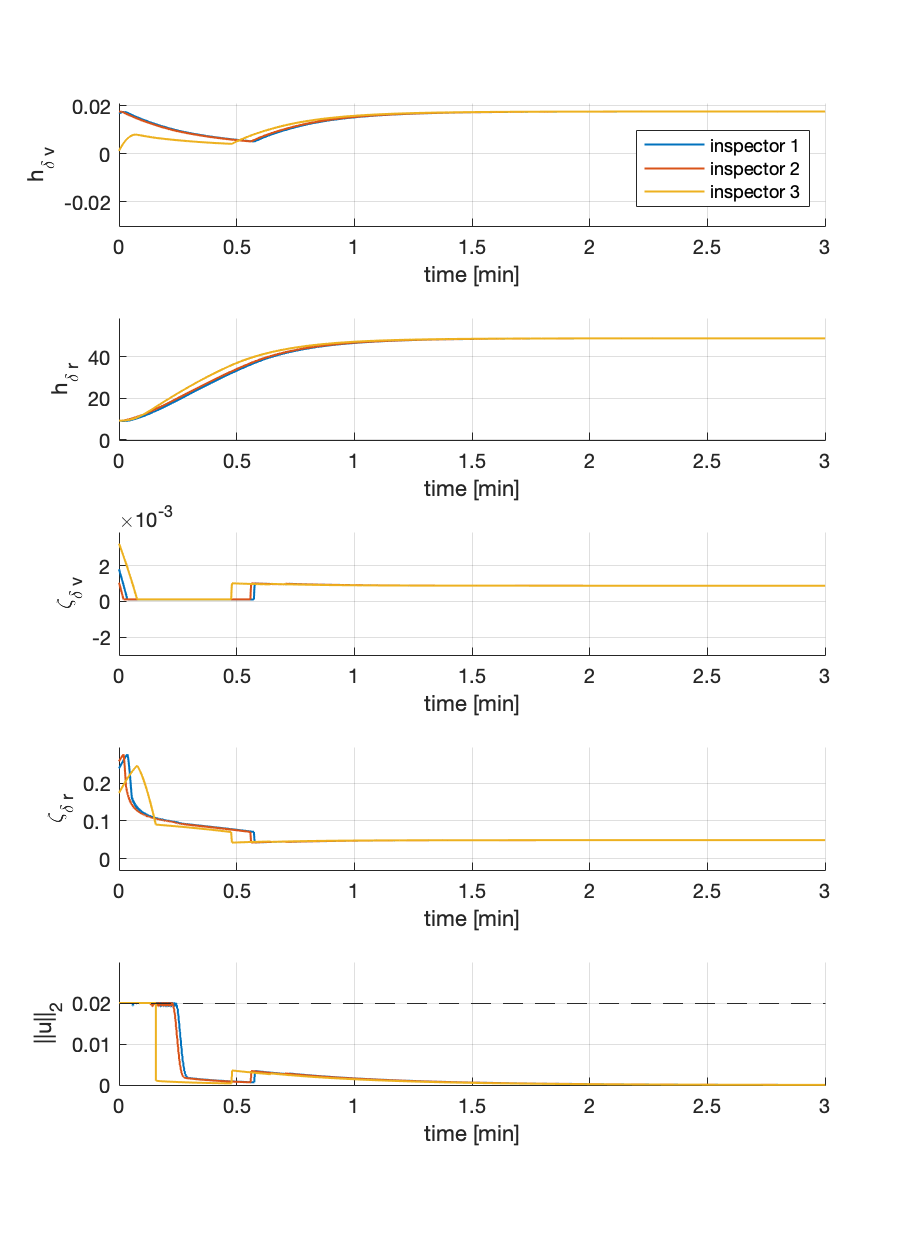}
    \caption{Time evolution of the CBFs $h_{\dvel}$,$h_{\dr}$; the HOCBF nominal function $\zeta_{\dr}$,$\zeta_{\dvel}$; and the control signal $\|\bm{u}\|_2$ for each inspector.}
    \label{fig:final barriers result}
    \vspace{-5mm}
\end{figure}

%% file: content/conclusion.tex
\begin{figure}
    \centering
\includegraphics[width=0.48\textwidth,trim={0cm 0.5cm 0 1cm},clip]{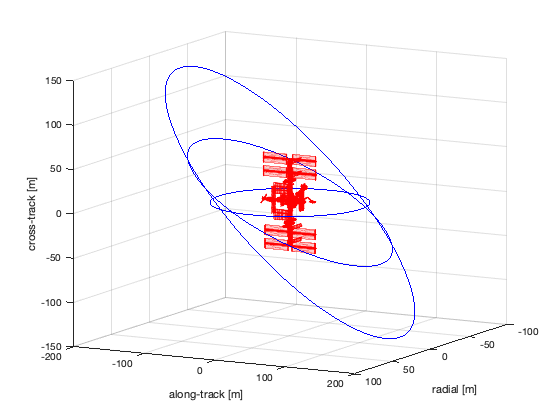}
    \caption{Three PRO assigned to each inspector around the ISS. The base directions are given in terms of the LVLH frame where $\hat{\bm{r}}$ is the \textit{radial} direction, $\hat{\bm{s}}$ is the \textit{along-track} direction and $\hat{\bm{w}}$ is the \textit{cross-track} direction.}
    \label{fig:PRO visualization}
    \vspace{-5mm}
\end{figure}
\label{conclusion}
In this work, we developed new definitions of sampled-data HOCBF thanks to which continuous-time safety guarantees can be ensured for sampled-data high-order systems. Then, we explored the applicability of the CMPC with the newly introduced definition of SD-HOCBF for a realistic space mission scenario. In future work, we will investigate how to determine suitable parameters for the CMPC scheme in a programmatic manner and we will analyse the impact of sensor noise on the control performance.

%% file: root.bbl
\begin{thebibliography}{10}
\providecommand{\url}[1]{#1}
\csname url@samestyle\endcsname
\providecommand{\newblock}{\relax}
\providecommand{\bibinfo}[2]{#2}
\providecommand{\BIBentrySTDinterwordspacing}{\spaceskip=0pt\relax}
\providecommand{\BIBentryALTinterwordstretchfactor}{4}
\providecommand{\BIBentryALTinterwordspacing}{\spaceskip=\fontdimen2\font plus
\BIBentryALTinterwordstretchfactor\fontdimen3\font minus \fontdimen4\font\relax}
\providecommand{\BIBforeignlanguage}[2]{{%
\expandafter\ifx\csname l@#1\endcsname\relax
\typeout{** WARNING: IEEEtran.bst: No hyphenation pattern has been}%
\typeout{** loaded for the language `#1'. Using the pattern for}%
\typeout{** the default language instead.}%
\else
\language=\csname l@#1\endcsname
\fi
#2}}
\providecommand{\BIBdecl}{\relax}
\BIBdecl

\bibitem{dorri_multi-agent_2018}
A.~Dorri, S.~S. Kanhere, and R.~Jurdak, ``Multi-{Agent} {Systems}: {A} {Survey},'' \emph{IEEE Access}, vol.~6, pp. 28\,573--28\,593, 2018.

\bibitem{mesbahi_graph_2010}
M.~Mesbahi and M.~Egerstedt, ``Graph theoretic methods in multiagent networks,'' in \emph{Graph {Theoretic} {Methods} in {Multiagent} {Networks}}.\hskip 1em plus 0.5em minus 0.4em\relax Princeton University Press, 2010.

\bibitem{ekblaw_self-assembling_2019}
A.~Ekblaw and J.~Paradiso, ``Self-{Assembling} {Space} {Habitats}: {TESSERAE} design and mission architecture,'' in \emph{{Aerospace} {Conference}}.\hskip 1em plus 0.5em minus 0.4em\relax IEEE, 2019, pp. 1--11.

\bibitem{khoshnevis_isru-based_2017}
B.~Khoshnevis, A.~Carlson, and M.~Thangavelu, ``{ISRU}-based robotic construction technologies for lunar and martian infrastructures,'' Tech. Rep., 2017.

\bibitem{flores-abad_review_2014}
A.~Flores-Abad, O.~Ma, K.~Pham, and S.~Ulrich, ``\BIBforeignlanguage{en}{A review of space robotics technologies for on-orbit servicing},'' \emph{\BIBforeignlanguage{en}{Progress in Aerospace Sciences}}, vol.~68, pp. 1--26, Jul. 2014.

\bibitem{nakka_information-based_2021}
Y.~K. Nakka, W.~H{\"o}nig, C.~Choi, A.~Harvard, A.~Rahmani, and S.-J. Chung, ``Information-based guidance and control architecture for multi-spacecraft on-orbit inspection,'' \emph{Journal of Guidance, Control, and Dynamics}, vol.~45, no.~7, pp. 1184--1201, 2022.

\bibitem{alfriend_spacecraft_2009}
K.~T. Alfriend, S.~R. Vadali, P.~Gurfil, J.~P. How, and L.~Breger, \emph{Spacecraft formation flying: {Dynamics}, control and navigation}.\hskip 1em plus 0.5em minus 0.4em\relax Elsevier, 2009, vol.~2.

\bibitem{zeng_safety-critical_2020}
J.~Zeng, B.~Zhang, and K.~Sreenath, ``Safety-critical model predictive control with discrete-time control barrier function,'' in \emph{American Control Conference (ACC)}.\hskip 1em plus 0.5em minus 0.4em\relax IEEE, 2021, pp. 3882--3889.

\bibitem{cortez_control_2019}
W.~S. Cortez, D.~Oetomo, C.~Manzie, and P.~Choong, ``Control barrier functions for mechanical systems: {Theory} and application to robotic grasping,'' \emph{IEEE Transactions on Control Systems Technology}, vol.~29, no.~2, pp. 530--545, 2019.

\bibitem{breeden_control_2021}
J.~Breeden, K.~Garg, and D.~Panagou, ``Control barrier functions in sampled-data systems,'' \emph{IEEE Control Systems Letters}, vol.~6, pp. 367--372, 2021.

\bibitem{roque_corridor_2022}
P.~Roque, W.~S. Cortez, L.~Lindemann, and D.~V. Dimarogonas, ``Corridor {MPC}: Towards optimal and safe trajectory tracking,'' in \emph{American Control Conference (ACC)}.\hskip 1em plus 0.5em minus 0.4em\relax IEEE, 2022, pp. 2025--2032.

\bibitem{tan_compatibility_2022}
X.~Tan and D.~V. Dimarogonas, ``Compatibility checking of multiple control barrier functions for input constrained systems,'' in \emph{Conference on Decision and Control (CDC)}.\hskip 1em plus 0.5em minus 0.4em\relax IEEE, 2022, pp. 939--944.

\bibitem{morgan_swarm-keeping_2012}
D.~Morgan, S.-J. Chung, L.~Blackmore, B.~Acikmese, D.~Bayard, and F.~Y. Hadaegh, ``Swarm-keeping strategies for spacecraft under {J2} and atmospheric drag perturbations,'' \emph{Journal of Guidance, Control, and Dynamics}, vol.~35, no.~5, pp. 1492--1506, 2012.

\bibitem{ames_control_2019}
A.~D. Ames, S.~Coogan, M.~Egerstedt, G.~Notomista, K.~Sreenath, and P.~Tabuada, ``Control barrier functions: {Theory} and applications,'' in \emph{{European} Control Conference ({ECC})}.\hskip 1em plus 0.5em minus 0.4em\relax IEEE, 2019, pp. 3420--3431.

\bibitem{xiao_control_2019}
W.~Xiao and C.~Belta, ``Control barrier functions for systems with high relative degree,'' in \emph{Conference on Decision and Control ({CDC})}.\hskip 1em plus 0.5em minus 0.4em\relax IEEE, 2019, pp. 474--479.

\bibitem{sontag_mathematical_2013}
E.~D. Sontag, \emph{Mathematical control theory: deterministic finite dimensional systems}.\hskip 1em plus 0.5em minus 0.4em\relax Springer Science \& Business Media, 2013, vol.~6.

\bibitem{scholtes_introduction_2012}
S.~Scholtes, \emph{Introduction to piecewise differentiable equations}.\hskip 1em plus 0.5em minus 0.4em\relax Springer Science \& Business Media, 2012.

\end{thebibliography}
